\documentclass[11pt]{article}

\usepackage[utf8]{inputenc}
\usepackage[T1]{fontenc}
\usepackage{microtype}
\usepackage{geometry}
\usepackage{latexsym}
\usepackage{amsmath,amssymb,amsthm,mathtools}
\usepackage[scr]{rsfso}
\usepackage{stmaryrd}
\usepackage{enumitem}
\usepackage{tikz-cd}
\usepackage{mdframed}
\usepackage{booktabs}
\usepackage{hyperref}
\hypersetup{
  colorlinks=true,
  linkcolor=blue,
  citecolor=blue,
  urlcolor=blue,
  pdftitle={Positive Topology and Feasible Refinement: Forcing Matrices, Positivity, and Information},
  pdfauthor={Mirco A. Mannucci; Giovanni Sambin},
  pdfsubject={Positive topology, constructive mathematics, forcing semantics, and feasible resource-aware refinement},
  pdfkeywords={positive topology, formal topology, pointfree topology, forcing semantics, Galois adjunctions, positivity, constructive mathematics, feasible refinement, resource-bounded semantics}
}

\theoremstyle{definition}
\newtheorem{definition}{Definition}[section]
\newtheorem{remark}[definition]{Remark}
\newtheorem{example}[definition]{Example}
\theoremstyle{plain}
\newtheorem{proposition}[definition]{Proposition}
\newtheorem{theorem}[definition]{Theorem}

\newcommand{\Pow}{\mathcal P}
\newcommand{\primconst}[1]{{\mathop{\sf #1}}}
\newcommand{\ext}{\primconst{\,ext\,}}
\newcommand{\rest}{\primconst{\,rest\,}}
\newcommand{\inte}{\primconst{\,int\,}}
\newcommand{\cl}{\primconst{\,cl\,}}
\newcommand{\A}{{\mathscr{A}}}
\newcommand{\J}{{\mathscr{J}}}

\newcommand{\forces}{\Vdash}
\newcommand{\cov}{\lhd}
\newcommand{\fish}{\ltimes}
\newcommand{\releps}{\mathrel{\varepsilon}}
\newcommand{\overlap}{\mathrel\between}
\newcommand{\sse}{\leftrightarrow}
\newcommand{\sub}{\subseteq}

\newcommand{\Split}{\mathsf{Split}}

\title{Positive Topology and Feasible Refinement: \\ Forcing Matrices, Positivity, and Information}
\author{Mirco A. Mannucci\\\texttt{mirco@holomathics.com}
\and Giovanni Sambin\\\texttt{sambin@math.unipd.it}}
\date{September 2026}

\begin{document}
\maketitle

\begin{center}


\medskip
\textbf{ACM CCS:} Theory of computation $\to$ Logic; Semantics and reasoning; Program semantics; Type theory; Software and its engineering $\to$ Software verification

\medskip
\textbf{arXiv categories:} math.LO (primary); math.CT, cs.LO (cross-list)
\end{center}

\begin{abstract}
Starting from a point--basic-open forcing relation $x\forces a$ between a set of points/models $X$
and basic opens/generators $S$, we obtain the canonical Galois adjunction
$\ext \dashv \Box$ between $\Pow(S)$ and $\Pow(X)$ encoding the ``open/locale''
(universal, cover/refinement) content~\cite{MacLaneMoerdijk1992}.
Less standard, but equally canonical, is a second adjunction $\Diamond \dashv \rest$ encoding the
``closed/positive'' (existential, probe/meet) content; this yields a direct origin for Sambin-style
positivity $a\fish U$ and ``closed without complements'' as fixed points of the contractive operator
$\J=\Diamond\circ\rest$~\cite{Sambin-OUP-PositiveTopology,CoquandSambinSmithValentini2003}.
We isolate a reconstruction principle: the forcing matrix is recovered uniquely from either residual pair via singleton tests.

When points are given, positivity can be induced from forcing in this way; in the formal pointfree approach,
positivity is instead taken as primitive and compatibility with cover is axiomatic.  This distinction supports
two operational readings~\cite{Sambin-OUP-PositiveTopology,CirauloVickers2016}: (i) an information-theoretic
view, where covers express refinement of partial information and positivity expresses witnessed feasibility,
and (ii) a game view, where a positive witness supplies a strategy through successive cover refinements.

The final part of the survey is deliberately programmatic.  We sketch how resource annotations may refine
ordinary forcing, how budget-indexed covers and positivity may then be studied, and how enriched variants
suggest resource-aware semantics.  We do \emph{not} claim here a final resource logic: the interaction among
budgeted forcing, cover, positivity, and logical connectives is deferred to a subsequent article~\cite{MannucciSambinForthcoming}.
Examples from medical diagnosis, legal reasoning, and AI systems illustrate the intended direction.
\end{abstract}

\tableofcontents

\section{Introduction: Why Positive Topology?}

Before diving into the formal machinery, we motivate the key ideas through a concrete scenario.

\subsection{A motivating example: Medical diagnosis}

Consider the task of diagnosing a patient based on observed symptoms.

\begin{itemize}
  \item $X$ = patients (the ``points'' or ``states of the world'')
  \item $S$ = symptoms, tests, conditions (the ``observables'' or ``basic opens'')
  \item $x \forces a$ means ``patient $x$ exhibits symptom/condition $a$''
\end{itemize}

A doctor observes that a patient has fever, cough, and fatigue. What can we infer?

\paragraph{The cover relation: semantic entailment.}
Certain conditions \emph{imply} certain symptoms. For instance:
\[
\mathsf{pneumonia} \cov \{\mathsf{fever}, \mathsf{cough}, \mathsf{chest\_pain}, \mathsf{infiltrates}\}
\]
This means: every patient with pneumonia exhibits at least one of fever, cough, chest pain, or lung infiltrates.
More generally, $a \cov U$ says: \emph{whenever $a$ holds, something in $U$ holds}.

This is the \textbf{universal/refinement} side: covers encode what implies what.

\paragraph{The positivity relation: witnessed consistency.}
Now consider the question: is the symptom cluster $\{\mathsf{fever}, \mathsf{cough}, \mathsf{fatigue}\}$ a \emph{realizable} profile?
That is, does there exist a patient whose \emph{complete} symptom profile is contained in this set?

We write:
\[
\mathsf{fever} \fish \{\mathsf{fever}, \mathsf{cough}, \mathsf{fatigue}\}
\]
to mean: there exists a patient with fever whose entire observable profile is contained in the displayed set.

This patient is a \textbf{witness} --- the symptom cluster is realizable, not a contradiction.

Contrast with:
\[
\mathsf{fever} \fish \{\mathsf{hypothermia}\}
\]
This is \emph{false}: no patient has both fever and only hypothermia. The combination is inconsistent.

This is the \textbf{existential/positivity} side: positivity encodes what is consistent, what has witnesses.

\paragraph{Why both matter.}
\begin{itemize}
  \item \textbf{Covers alone} tell us what implies what, but not what's realizable.
  \item \textbf{Positivity alone} tells us what's consistent, but not what follows from what.
  \item \textbf{Together} they give us: grounded inference (every conclusion has a witness) with semantic structure (implications are tracked).
\end{itemize}

\paragraph{The compatibility axiom.}
The key interaction between covers and positivity is:
\[
(a \cov U) \land (a \fish V) \implies \exists u \in U.\; u \fish V
\]
In medical terms: if pneumonia implies \{fever, cough, infiltrates\}, and there exists a pneumonia patient confined to profile $V$,
then there exists a fever patient (or cough patient, or infiltrates patient) also confined to $V$.

\textbf{Refinement preserves witnesses.} This is the soundness condition that makes the framework coherent.

\subsection{The general pattern}

The medical example instantiates a universal pattern:

\begin{center}
\begin{tabular}{lll}
\toprule
\textbf{Domain} & \textbf{Points $X$} & \textbf{Observables $S$} \\
\midrule
Medicine & Patients & Symptoms, tests, conditions \\
Law & Cases & Evidence, claims, precedents \\
Programming & Program states & Properties, assertions, tests \\
Finance & Portfolios & Risk factors, exposures \\
Science & Hypotheses & Experiments, observations \\
AI agents & Traces/executions & Constraints, evaluations \\
\bottomrule
\end{tabular}
\end{center}

This operational reading follows the constructive-positivity tradition initiated by Sambin~\cite{Sambin-OUP-PositiveTopology}
and formalized via inductively generated formal topologies~\cite{CoquandSambinSmithValentini2003}.
It abstracts the familiar locale/sheaf perspective of Mac~Lane--Moerdijk~\cite{MacLaneMoerdijk1992} down to the level of the forcing matrix we use in code.

In each domain:
\begin{itemize}
  \item $a \cov U$ encodes \emph{semantic entailment}: $a$ implies something in $U$
  \item $a \fish U$ encodes \emph{witnessed consistency}: there exists $x$ with $x\forces a$ and $\Diamond x\subseteq U$
  \item The forcing $x \forces a$ connects the abstract observables to concrete states
\end{itemize}

The rest of this paper develops these ideas formally through the forcing/adjunction/positivity construction.
The resource-aware part is intentionally an \emph{esquisse d'un programme}: it records a plausible direction and
isolates questions whose full treatment is left to the subsequent article~\cite{MannucciSambinForthcoming}.

\section{The starting datum: a forcing/incidence relation}

\begin{definition}[Forcing table]
Fix sets $X$ (points/models) and $S$ (basic opens/generators).
A \emph{forcing relation} is a subset $R\subseteq X\times S$ written
\[
x\forces a \quad\text{for }(x,a)\in R.
\]
\end{definition}

\begin{remark}[Interpretation]
Think of $a\in S$ as a basic open (or observable, or test) and $x\in X$ as a point (or state, or model).
Then $x\forces a$ reads: ``$x$ lies in the basic open $a$'' or ``$x$ satisfies observable $a$''.
No topology is assumed at the outset: we start from the incidence matrix.
\end{remark}

\begin{example}[Medical diagnosis]\label{ex:medical-forcing}
Let $X$ be a set of patients and $S$ a set of symptoms/conditions.
Then $x \forces a$ means ``patient $x$ exhibits symptom $a$''.
The forcing relation is the patient-symptom incidence matrix.
\end{example}

\begin{example}[Program verification]
Let $X$ be program states and $S$ be assertions/properties.
Then $x \forces a$ means ``state $x$ satisfies property $a$''.
\end{example}

\begin{remark}[Boolean matrix viewpoint]
The relation $R\subseteq X\times S$ is a Boolean-valued matrix
$R:X\times S\to \{0,1\}$.
Much of what follows is ``linear algebra over $\{0,1\}$'': joins are unions, meets are intersections.
This matrix viewpoint also suggests the enriched, resource-sensitive direction sketched in Section~\ref{sec:resources}.
\end{remark}

\section[The first adjunction: opens via ext adjoint Box]{The first adjunction: opens via $\ext \dashv \Box$}\label{sec:adj1}

\subsection{Direct image and the universal residual}

\begin{definition}[Extension and interior-kernel]
Define monotone maps
\[
\ext:\Pow(S)\to\Pow(X),\qquad \ext U=\{x\in X:\exists a\in U,\ x\forces a\},
\]
\[
\Box:\Pow(X)\to\Pow(S),\qquad \Box E=\{a\in S:\forall x\in X,\ (x\forces a\Rightarrow x\in E)\}.
\]
\end{definition}

\begin{example}[Medical: Extension and Interior]
Continuing Example~\ref{ex:medical-forcing}:
\begin{itemize}
  \item $\ext \{\mathsf{fever}, \mathsf{cough}\}$ = patients who have fever OR cough (or both)
  \item $\Box E$ where $E$ = ``patients over 65'' gives: symptoms that \emph{only} occur in patients over 65
\end{itemize}
\end{example}

\begin{proposition}[Galois adjunction I]
For all $U\subseteq S$ and $E\subseteq X$,
\[
\ext U\subseteq E \quad\Longleftrightarrow\quad U\subseteq \Box E.
\]
Equivalently, $\ext\dashv \Box$.
\end{proposition}

\begin{proof}
($\Rightarrow$) If $a\in U$, then $\ext \{a\}\subseteq \ext U\subseteq E$, hence $a\in \Box E$.
($\Leftarrow$) If $x\in \ext U$, choose $a\in U$ with $x\forces a$. Since $a\in\Box E$, we have $x\in E$.
\end{proof}

\begin{remark}[Universal nature]
$\Box E$ is defined by a $\forall$-condition; this is the ``universal/refinement'' side.
This is exactly why the cover relation induced by forcing is a universal statement.
\end{remark}

\subsection{Cover and the pointfree (formal-cover) side}

\begin{definition}[Induced preorder and cover]
Define a preorder on $S$ by
\[
a\le b \quad:\Longleftrightarrow\quad \forall x,\ (x\forces a\Rightarrow x\forces b),
\]
and define a cover relation $\cov$ by
\[
a\cov U \quad:\Longleftrightarrow\quad \forall x,\ (x\forces a \Rightarrow \exists b\in U,\ x\forces b).
\]
\end{definition}

\begin{example}[Medical: Cover as symptom entailment]\label{ex:medical-cover}
In the medical setting, $a \cov U$ means: every patient with condition $a$ exhibits at least one symptom in $U$.
For instance:
\begin{align*}
\mathsf{flu} &\cov \{\mathsf{fever}, \mathsf{cough}, \mathsf{fatigue}, \mathsf{myalgia}\} \\
\mathsf{bacterial\_infection} &\cov \{\mathsf{fever}, \mathsf{elevated\_WBC}\}
\end{align*}
These are \emph{diagnostic implications}: the disease implies certain observable symptoms.
\end{example}

\begin{remark}[Extensional reading]
Equivalently:
\[
a\le b \iff \ext \{a\}\subseteq \ext \{b\},\qquad
a\cov U \iff \ext \{a\}\subseteq \ext U.
\]
So $\cov$ is literally ``inclusion in a union''.
\end{remark}

\subsection{Extensional topology on points}

\begin{definition}[Generated opens on $X$]
A subset $E\subseteq X$ is \emph{extensional open} (generated by $S$) if
\[
E=\ext \Box E.
\]
\end{definition}

\begin{remark}
The operator $\inte:=\ext\circ\Box$ is an \emph{interior} operator on $\Pow(X)$: it extracts the largest extensional open contained in a given subset.
(Dually, $\A:=\Box\circ\ext$ is the \emph{saturation} operator, a closure/nucleus-like operator on $\Pow(S)$.)
Extensional opens are exactly the fixed points of $\inte$: ``those subsets describable by generators''.
\end{remark}

\begin{definition}[Extensional equivalence of points]
Define $x\sim y$ if the induced neighborhood sets agree:
\[
x\sim y \quad:\Longleftrightarrow\quad \forall a\in S,\ (x\forces a \Leftrightarrow y\forces a).
\]
\end{definition}

\begin{example}[Medical: Indistinguishable patients]
Two patients $x \sim y$ are \emph{symptomatically indistinguishable}: they exhibit exactly the same symptoms.
From a diagnostic standpoint, they are equivalent --- the same inferences apply to both.
\end{example}

\begin{remark}
Pointfree semantics does not distinguish points inside the same $\sim$-class.
This is the basic ``quotient-to-locale'' phenomenon: locale data sees neighborhood behavior, not point multiplicity.
\end{remark}

\section[The second adjunction: positivity via Diamond adjoint rest]{The second adjunction: positivity via $\Diamond \dashv \rest$}\label{sec:adj2}

The locale/cover side is universal; ``closed without complements'' needs an existential probe:
``does a basic open \emph{hit} a closed thing?''.
This is the conceptual origin of Sambin-style positivity~\cite{Sambin-OUP-PositiveTopology,Valentini2012,CirauloVickers2016}.

\subsection{Neighborhood sets and constraint-selected points}

\begin{definition}[Neighborhood set]
For $x\in X$, its neighborhood set (or \emph{observable profile}) is
\[
\Diamond x:=\{a\in S:\ x\forces a\}\subseteq S.
\]
\end{definition}

\begin{example}[Medical: Symptom profile]
For a patient $x$, the neighborhood set $\Diamond x$ is their \emph{complete symptom profile}:
the set of all symptoms/conditions they exhibit.
\end{example}

\begin{definition}[Selection by a constraint on neighborhoods]
Define
\[
\rest:\Pow(S)\to\Pow(X),\qquad
\rest U:=\{x\in X:\ \Diamond x\subseteq U\}.
\]
Thus $\rest U$ consists of points all of whose basic neighborhoods lie in $U$.
\end{definition}

\begin{example}[Medical: Patients confined to a symptom cluster]
$\rest \{\mathsf{fever}, \mathsf{cough}, \mathsf{fatigue}\}$ = patients whose \emph{entire} symptom profile
is contained in $\{\mathsf{fever}, \mathsf{cough}, \mathsf{fatigue}\}$.

These are patients who have \emph{only} fever, cough, and/or fatigue --- no other symptoms.
\end{example}

\subsection{The existential ``hit'' map}

\begin{definition}[$\Diamond$-Hit]
Define
\[
\Diamond:\Pow(X)\to\Pow(S),\qquad
\Diamond D:=\{a\in S:\exists x\in D,\ x\forces a\},
\qquad\text{that is,}\quad a\releps\Diamond D \;\sse\; \ext a\overlap D.
\]
\end{definition}
In words, a hits the set $D$ if $a$ has a witness inside $D$.

\begin{example}[Medical: Symptoms witnessed in a cohort]
If $D$ = patients in a clinical trial, then $\Diamond D$ = symptoms exhibited by at least one patient in the trial.
\end{example}

\begin{proposition}[Galois adjunction II]
For all $D\subseteq X$ and $U\subseteq S$,
\[
\Diamond D\subseteq U \quad\Longleftrightarrow\quad D\subseteq \rest U.
\]
Equivalently, $\Diamond\dashv \rest$.
\end{proposition}

\begin{proof}
($\Rightarrow$) Let $x\in D$. If $a\in \Diamond x$ then $a\in\Diamond D\subseteq U$. Hence $\Diamond x\subseteq U$ and $x\in\rest U$.

($\Leftarrow$) If $a\in\Diamond D$, choose $x\in D$ with $x\forces a$. Since $x\in\rest U$ we have $a\in \Diamond x\subseteq U$.
\end{proof}

\begin{remark}[Meaning]
$\Diamond D$ records which generators are \emph{witnessed} in $D$.
$\rest U$ picks points whose \emph{entire neighborhood behavior} is constrained by $U$.
This is the ``closed-without-complements'' move: $U$ acts as a \emph{code} for a closed-like object by constraining neighborhoods, not by negation.
\end{remark}

\section{Positivity and closed sets without complements}

\subsection[The positivity operator J = Diamond o rest]{The positivity operator $\J=\Diamond\circ\rest$}

\begin{definition}[Positivity interior on codes]
Define
\[
\J:\Pow(S)\to\Pow(S),\qquad \J U:=\Diamond \rest U.
\]
Equivalently,
\[
a\releps \J U\quad\sse\quad \exists x\in X\ \bigl(x\forces a\ \wedge\ \Diamond x\subseteq U\bigr)
\quad\sse\quad \ext a \overlap \rest U.
\]
\end{definition}

\begin{definition}[Positivity relation]
Define $\fish\subseteq S\times \Pow(S)$ by
\[
a\fish U \quad:\Longleftrightarrow\quad a\in \J U.
\]
\end{definition}

\begin{remark}[One-line intuition]
$a\fish U$ means: there exists a point $x$ in basic open $a$ such that all basic neighborhoods of $x$ lie in $U$.
So $U$ behaves like a ``closed neighborhood specification'' and $a$ \emph{hits} it.
\end{remark}

\begin{example}[Medical: Positivity as witnessed symptom clusters]\label{ex:medical-positivity}
The statement
\[
\mathsf{fever} \fish \{\mathsf{fever}, \mathsf{cough}, \mathsf{fatigue}\}
\]
means that there exists a patient who has fever and whose \emph{complete} symptom profile is contained in that set.

This patient is a \textbf{witness} that fever is consistent with the symptom cluster.

Contrast: $\mathsf{fever} \fish \{\mathsf{hypothermia}\}$ is \emph{false} because no patient
has both fever and only hypothermia --- the combination is contradictory.

Positivity detects \textbf{realizability}: which symptom combinations actually occur in patients.
\end{example}

\begin{proposition}[A point forces positivity]\label{prop:point-implies-positive}
Assume a point $x\in X$ and a generator $a\in S$ with $x\forces a$.
Then $a$ is positive in the unary sense, i.e.\ $a\fish S$.
More generally, if $U\subseteq S$ satisfies $\Diamond x\subseteq U$, then $a\fish U$.
\end{proposition}

\begin{proof}
If $\Diamond x\subseteq U$ then $x\in\rest U$ by definition, hence $a\in\Diamond \rest U=\J U$,
i.e.\ $a\fish U$.  Taking $U=S$ gives $a\fish S$.
\end{proof}

\subsection{Elementary properties}

\begin{proposition}[Contractive and monotone]
For all $U,V\subseteq S$:
\begin{enumerate}[label=(\alph*)]
\item $\J U\subseteq U$.
\item If $U\subseteq V$ then $\J U\subseteq \J V$.
\end{enumerate}
\end{proposition}

\begin{proof}
(a) If $a\in \J U$, choose $x$ with $x\forces a$ and $\Diamond x\subseteq U$. Then $a\in \Diamond x\subseteq U$.

(b) If $U\subseteq V$ then $\rest U\subseteq \rest V$, hence $\Diamond \rest U\subseteq \Diamond \rest V$.
\end{proof}

\begin{proposition}[Idempotence]\label{prop:idempotence-laws}
$\J= \Diamond\circ\rest$ is idempotent: $\J \J U=\J U$.
\end{proposition}

\begin{proof}
Since $\J\le \mathrm{id}$, we have $\J \J U\subseteq \J U$.
For the reverse, $\mathrm{id}\le \rest\circ\Diamond$ implies $\rest U\subseteq \rest \J U$; applying $\Diamond$ yields $\J U\subseteq \J \J U$.
\end{proof}

\begin{definition}[Formal closed (via positivity)]
Call $U\subseteq S$ \emph{formal closed} if $\J U=U$.
\end{definition}

\begin{example}[Medical: Complete symptom profiles]\label{ex:medical-closed}
Because $\J U\subseteq U$ holds for every $U$, the fixed-point condition $\J U=U$ amounts to the converse inclusion
$U\subseteq \J U$.  Thus a set $U\subseteq S$ is formal closed iff, for every symptom $a\in U$, there exists a patient
$x$ such that
\[
x\forces a \qquad\text{and}\qquad \Diamond x\subseteq U.
\]
In other words, every observable admitted by the closed code $U$ is actually realized by a witness whose complete
observable profile stays inside $U$.  This positive formulation is the useful constructive content of closedness here.
\end{example}

\subsection{The compatibility axiom}

The key interaction between covers and positivity is captured by:

\begin{proposition}[Compatibility]\label{prop:compatibility}
If $\A, \J$ are defined as above from $\forces$, then for all $U, V \subseteq S$:
\[
\A U \overlap \J V \;\to\; U \overlap \J V.
\]
\end{proposition}

\begin{proof}
Suppose $\A U \overlap \J V$, that is, there exists $a$ with $a \releps \A U$ and $a \fish V$.
By positivity ($a\fish V$), there exists $x$ with $x \forces a$ and $\Diamond x \subseteq V$.
By $a \releps \A U$, i.e.\ $\ext a \sub \ext U$, since $x \forces a$ there exists $u \in U$ with $x \forces u$.
Then $u \in \Diamond x \subseteq V$, so $u \fish V$; hence $U \overlap \J V$.
\end{proof}

Compatibility is provable when cover and positivity are defined via points.  In a formal approach, it is assumed as an axiom. 

\begin{remark}[Why compatibility matters]
Compatibility is the \textbf{soundness condition} for refinement:
if you refine a hypothesis (via covers), you don't lose your witnesses (positivity).
This is what makes the formal framework coherent for inference.
\end{remark}

\begin{example}[Medical: Refinement preserves witnesses]
Suppose
\[
\mathsf{pneumonia} \cov \{\mathsf{fever}, \mathsf{cough}, \mathsf{infiltrates}\},
\]
and there exists a pneumonia patient confined to profile $V$.

Then there exists a fever patient (or cough patient, or infiltrates patient) also confined to $V$.

The diagnostic refinement from ``pneumonia'' to its symptoms preserves the witness.
\end{example}

\section{Information content and refinement games}\label{sec:info-games}

The formal relations $\cov$ and $\fish$ admit an interpretation closer to computation than to classical topology.
Think of a generator $a\in S$ as a \emph{partial information state}.  A cover $a\cov U$
says that to \emph{settle} the information in $a$ one must proceed to a feasible refinement and end up in \emph{one of} the subcases in $U$.

Positivity adds one crucial twist: existence is not classical negation of emptiness, but a \emph{witnessed} possibility of continuing the refinement process.

\subsection[The refinement game G(a,U)]{The refinement game $G(a,U)$}

Fix $a\in S$ and a set $U\subseteq S$ of ``acceptable outcomes''.
Define an infinite two-player game $G(a,U)$:

\begin{itemize}
  \item The current position is a generator $b\in S$ (initially $b=a$).
  \item \emph{Refiner} (``Devil'') chooses a cover $b\cov V$.
  \item \emph{Witness} (``God'') responds by choosing some $v\in V$.
  \item The next position is $b:=v$.
  \item Witness wins if she can play forever while keeping every chosen position inside $U$.
\end{itemize}

\begin{proposition}[Positivity as a winning condition]\label{prop:positivity-game}
If $a\fish U$, then Witness has a winning strategy in $G(a,U)$.
\end{proposition}

\begin{proof}
By $a\in \J U$ there exists $x\in X$ with $x\forces a$ and $\Diamond x\subseteq U$.
Whenever Refiner plays $b\cov V$ with $x\forces b$, there exists $v\in V$ with $x\forces v$.
Witness chooses such $v$. Since $v\in \Diamond x\subseteq U$, the position stays in $U$.
\end{proof}

\paragraph{Medical diagnosis as a refinement game.}
The game $G(\mathsf{fever}, U)$ where $U$ is a diagnostic profile:
\begin{itemize}
  \item Refiner (skeptic): ``You say fever? What kind? Bacterial? Viral?''
  \item Witness (diagnostician): ``Bacterial fever'' (chooses refinement in $U$)
  \item Refiner: ``What symptoms? Give me a cover.''
  \item Witness: ``Elevated WBC'' (still in $U$)
\end{itemize}
A real patient $x$ witnessing positivity supplies Witness with a strategy that survives every refinement challenge.

\paragraph{Legal reasoning as a refinement game.}
Let $X$ = legal cases, $S$ = claims/evidence/precedents.
The game: opposing counsel (Refiner) demands finer distinctions; the advocate (Witness) must find precedents.
Positivity = having a real case to cite at every challenge.

\section{Two adjunctions and four operators}\label{sec:two-adjunctions}

The same forcing table gives two adjunctions, hence four operators:

\[
\Pow(S)\ \underset{\Box}{\overset{\ext}{\rightleftarrows}}\ \Pow(X)
\qquad(\ext\dashv \Box),
\qquad\qquad
\Pow(X)\ \underset{\rest}{\overset{\Diamond}{\rightleftarrows}}\ \Pow(S)
\qquad(\Diamond\dashv \rest).
\]

\begin{center}
\begin{tabular}{lll}
\toprule
\textbf{Operator} & \textbf{Type} & \textbf{Meaning} \\
\midrule
$\ext$ & $\Pow(S) \to \Pow(X)$ & Points satisfying some observable in $U$ \\
$\Box$ & $\Pow(X) \to \Pow(S)$ & Observables universally valid on $E$ \\
$\Diamond$ & $\Pow(X) \to \Pow(S)$ & Observables witnessed in $D$ \\
$\rest$ & $\Pow(S) \to \Pow(X)$ & Points confined to profile $U$ \\
\midrule
$\A = \Box \circ \ext$ & $\Pow(S) \to \Pow(S)$ & Saturation (formal opens) \\
$\inte = \ext \circ \Box$ & $\Pow(X) \to \Pow(X)$ & Interior (extensional opens) \\
$\J = \Diamond \circ \rest$ & $\Pow(S) \to \Pow(S)$ & Reduction (formal closed codes) \\
$\cl = \rest \circ \Diamond$ & $\Pow(X) \to \Pow(X)$ & Closure (on points) \\
\bottomrule
\end{tabular}
\end{center}

\begin{remark}[Two sides of the same coin]
Both adjunctions arise from the same forcing relation, but govern different semantics:
\begin{itemize}
\item $\ext\dashv\Box$ = universal refinement / locale / covers
\item $\Diamond\dashv\rest$ = existential probes / positivity / closed-without-complements
\end{itemize}
\end{remark}

\section{Reconstruction of forcing from residuals}

A central observation (cf.\ Ciraulo--Sambin~\cite{CirauloSambin2012}) is that the forcing table is recoverable from either left adjoint,
since left adjoints preserve joins and are determined by their action on singletons.

\begin{theorem}[Singleton reconstruction]
If $F=\ext$ arises from $\forces$, then $\forces$ is uniquely determined by:
\[
x\forces a \Longleftrightarrow x\in \ext \{a\}.
\]
Similarly for $\Diamond$:
\[
x\forces a \Longleftrightarrow a\in \Diamond \{x\}.
\]
\end{theorem}

\begin{proof}
Since $\ext$ preserves unions, $\ext U=\bigcup_{a\in U}\ext \{a\}$.
Thus $x\in \ext U$ iff $\exists a\in U$ with $x\in \ext \{a\}$.
\end{proof}

\begin{remark}[What is \emph{not} reconstructible]
The composites $j=\Box\circ\ext$ or $\J=\Diamond\circ\rest$ generally lose witness-level information.
The residual \emph{pair} determines forcing; a single composite does not.
\end{remark}

\section{Feasible refinement and resources: a programmatic sketch}\label{sec:resources}

The preceding sections are the mathematically established core of this survey.  We now turn to resources.
This section is intentionally exploratory: its purpose is to identify a clean baseline and the questions that
must be settled in a fuller theory of resource-bounded Positive Topology~\cite{MannucciSambinForthcoming}.

\subsection{Truth first, cost second: a minimal baseline}\label{subsec:ordered-costs}

A resource annotation should refine ordinary forcing rather than replace truth by affordability.  We therefore
retain the Boolean forcing relation $x\forces a$ and annotate verified incidences with a cost.

\begin{definition}[Budget preorder]
A \emph{budget preorder} is a preorder $(B,\preceq)$.  We read $b\preceq b'$ as saying that anything
feasible within budget $b$ is also feasible within the more permissive budget $b'$.
\end{definition}

\begin{definition}[Cost annotation and budgeted forcing]\label{def:cost-valued-forcing}
Let
\[
c:\{(x,a)\in X\times S:x\forces a\}\longrightarrow B
\]
assign a verification cost to each true incidence.  For $b\in B$ define
\[
x\forces_b a
\quad:\!\!\Longleftrightarrow\quad
x\forces a\ \text{ and }\ c(x,a)\preceq b.
\]
Thus $x\forces_b a$ means not merely that testing $a$ costs at most $b$, but that $a$ is true at $x$ and
that this truth can be verified within budget $b$.
\end{definition}

\begin{example}[Medical: Verification costs]\label{ex:medical-costs}
For illustration, suppose verified observations carry the following costs:

\begin{center}
\begin{tabular}{lrl}
\toprule
\textbf{Test/Observation} & \textbf{Cost} & \textbf{Notes} \\
\midrule
Fever & \$0 & Thermometer \\
Cough & \$0 & Observation \\
Blood pressure & \$5 & Standard equipment \\
Blood test (WBC) & \$50 & Lab work \\
PCR test & \$100 & Molecular \\
X-ray & \$200 & Imaging \\
CT scan & \$500 & Advanced imaging \\
\bottomrule
\end{tabular}
\end{center}

If patient $x$ really has elevated WBC and that fact costs \$50 to verify, then $x\forces_b\mathsf{elevated\_WBC}$
for every budget $b$ at least \$50.  A negative test result is not represented by a cheap positive incidence;
it belongs to the underlying observational vocabulary only if we explicitly include a corresponding generator.
\end{example}

\begin{remark}[Filtration of forcing]
For $b\preceq b'$ we have
\[
x\forces_b a\Longrightarrow x\forces_{b'}a.
\]
Hence the \emph{forcing relations} form a monotone filtration in the budget.  No corresponding monotonicity of
induced covers or positivity relations is asserted here.
\end{remark}

\subsection{Budget-indexed covers and positivity: the next problem}

At each fixed budget $b$ one may repeat the constructions of the Boolean case:
\[
a\cov_b U
\quad:\Longleftrightarrow\quad
\forall x\,\bigl(x\forces_b a\Rightarrow \exists u\in U\;x\forces_bu\bigr),
\]
\[
a\fish_b U
\quad:\Longleftrightarrow\quad
\exists x\,\bigl(x\forces_ba\ \land\ \Diamond_bx\subseteq U\bigr),
\qquad
\Diamond_bx:=\{a\in S:x\forces_ba\}.
\]
For every fixed $b$, these are simply the cover and positivity relations induced by the budget-restricted forcing table.

The subtle question is how they vary with $b$.  Although $\forces_b$ is monotone in the budget, neither
$\cov_b$ nor $\fish_b$ inherits a naive monotonicity law automatically.  Increasing a budget may introduce
new points satisfying the antecedent of a cover, new possible witnesses, and a larger observable profile
$\Diamond_bx$; these effects need not point in the same direction.  Characterizing conditions under which
families $\{\cov_b\}_{b\in B}$ and $\{\fish_b\}_{b\in B}$ have useful coherence properties is one of the
main tasks deferred to the next article~\cite{MannucciSambinForthcoming}.

\subsection{Enriched values: an outlook, not a final semantics}

A more ambitious route is to replace the Boolean matrix by a value object $V$ carrying both logical and resource
information.  Schematically one seeks
\[
R:X\times S\longrightarrow V
\]
together with suitable projections or threshold maps recovering ordinary and budget-indexed forcing.  In particular,
the Boolean content should remain recoverable from the enrichment rather than being identified with a numerical cost.

If $V$ is equipped with a sufficiently complete residuated monoidal order $(V,\le,\otimes,I,\multimap)$, the familiar
enriched formulas suggest the operators
\[
(\ext_V\phi)(x)=\bigvee_{a\in S}\phi(a)\otimes R(x,a),
\qquad
(\Box_V\psi)(a)=\bigwedge_{x\in X}\bigl(R(x,a)\multimap\psi(x)\bigr),
\]
\[
(\Diamond_V\psi)(a)=\bigvee_{x\in X}\psi(x)\otimes R(x,a),
\qquad
(\rest_V\phi)(x)=\bigwedge_{a\in S}\bigl(R(x,a)\multimap\phi(a)\bigr).
\]
Under the usual completeness/residuation hypotheses these have the expected adjoint shape
$\ext_V\dashv\Box_V$ and $\Diamond_V\dashv\rest_V$, and suggest an enriched positivity operator
$\J_V:=\Diamond_V\circ\rest_V$.

We record these formulas only as a bridge to the programme.  The precise choice of $V$, the recovery of the
underlying truth relation, the role of thresholds, and the interaction with logical structural rules are not fixed
in this survey and will be treated separately~\cite{MannucciSambinForthcoming}.

\section{Conceptual summary and outlook}\label{sec:summary-outlook}

\begin{center}
\fbox{\begin{minipage}{0.9\linewidth}
\textbf{Three layers, with different status in this survey.}
\[
\begin{array}{c}
\text{(i) forcing and the two adjunctions}\\[2pt]
\Downarrow\\[2pt]
\text{(ii) cover, positivity, and closedness without complements}\\[2pt]
\Downarrow\\[2pt]
\text{(iii) resource-aware refinement and logic}
\end{array}
\]
Layers (i)--(ii) are developed here.  Layer (iii) is the announced research programme.
\end{minipage}}
\end{center}

\subsection{Outlook: toward budget-indexed truth}

A future resource logic should introduce a new satisfaction judgement, say
\[
s\Vdash_b\varphi,
\]
where $s$ is an explicitly specified semantic state, $\varphi$ belongs to a defined syntax, and $b$ is a budget.
This judgement is \emph{not} the same object as the incidence relation $x\forces_b a$ of Section~\ref{sec:resources};
it would be a recursively defined logical semantics built on top of the resource-refined forcing structure.

For example, if resources must be split between conjunctive subgoals, one might require a relation
$\Split(b;b_1,b_2)$ and adopt a clause of the schematic form
\[
s\Vdash_b(\varphi\wedge\psi)
\quad\Longleftrightarrow\quad
\exists b_1,b_2\,
\bigl(\Split(b;b_1,b_2)\ \&\ s\Vdash_{b_1}\varphi\ \&\ s\Vdash_{b_2}\psi\bigr).
\]
Likewise, existential truth should remain witnessed and budget-sensitive.  But the syntax, semantic states,
budget algebra, structural rules, and exact relation between logical satisfaction and positivity are deliberately
left open here.  They form part of the next article~\cite{MannucciSambinForthcoming}.

\appendix
\section{Application to AI engineering}

This appendix explains why forcing/positivity/cost provides concrete engineering value for AI systems.

\subsection{The AI analogue of forcing}

In AI engineering, one evaluates systems against requirements:
\begin{itemize}
\item $X$ = executions / traces / agent rollouts
\item $S$ = constraints (``tool-call succeeded'', ``no disallowed content'', ``latency $<500$ms'', etc.)
\item $x\forces a$ = ``execution $x$ satisfies constraint $a$''
\end{itemize}

The Boolean forcing table $R\subseteq X\times S$ is therefore an evaluation matrix.  The resource-sensitive
constructions of Section~\ref{sec:resources} should be read as possible refinements of this basic evaluation layer.

\subsection{Two adjunctions = two engineering operations}

\paragraph{(1) $\ext\dashv\Box$: specification refinement.}
\begin{itemize}
\item $\ext U$: which traces satisfy at least one check in $U$?
\item $\Box E$: which checks are universally valid over traces $E$?
\item $a\cov U$: whenever $a$ holds, some $b\in U$ holds
\end{itemize}
This is specification management: what the system guarantees.

\paragraph{(2) $\Diamond\dashv\rest$: witness finding.}
\begin{itemize}
\item $\Diamond D$: which checks are witnessed in trace-set $D$?
\item $\rest U$: which traces live entirely inside constraint envelope $U$?
\item $a\fish U$: exists a trace satisfying $a$ with profile $\subseteq U$
\end{itemize}
This is counterexample search and safe operating region identification.

\subsection{Positivity = safety envelopes}

Engineers want ``safe = complement of bad''. This fails because:
\begin{enumerate}[label=(\alph*)]
\item the complement is not observable (unknown unknowns)
\item violations are sparse and adversarial
\item safety is better defined by positive evidence of robustness
\end{enumerate}

Positivity gives a constructive alternative:
\begin{itemize}
\item closed code $U$ with $\J U=U$ is a safety envelope
\item $a\fish U$ provides witness-based membership
\end{itemize}

Safety is a robust region, not a complement.

\subsection{LLM + PosTop: an illustrative architecture}\label{subsec:llm-postop}

Current LLM agents can fail by producing claims without witnesses, by combining inconsistent claims, or by spending
verification resources poorly.  The forcing/cover/positivity picture suggests an engineering discipline in which
candidate outputs are checked against explicit observations, witnesses, and budgets.

At this point, however, a typing distinction is essential.  The present survey has defined cover and positivity for
generators $a\in S$ and subsets $U\subseteq S$; it has \emph{not} identified arbitrary LLM hypotheses with generators.
A concrete system must therefore provide an encoding of hypotheses into the formal structure -- for example a map
$\iota:H\to S$ or a richer syntax built over $S$ -- together with a representation of the current observation state.
Only after that choice is made do expressions involving cover or positivity become well typed.

One possible implementation pattern is then:
\begin{enumerate}
  \item the LLM proposes candidate hypotheses $H$;
  \item an explicit encoding maps each candidate into the PosTop structure;
  \item cover checks test the chosen direction of semantic consequence, while positivity checks demand a witness;
  \item a resource layer proposes an observation or test that best discriminates among the remaining candidates;
  \item the process repeats until the task is resolved or the budget is exhausted.
\end{enumerate}

This is an architectural \emph{schema}, not a theorem of the present paper.  In particular, the direction of the cover
check depends on how an observation state and a hypothesis are represented; the informal expression
``$h\cov\mathrm{observed}$'' used in earlier drafts should not be read as defined notation.  A fully typed account of
hypotheses, observations, resource choice, and agent traces belongs to the subsequent development~\cite{MannucciSambinForthcoming}.

\paragraph{Comparison with knowledge graphs.}

The comparison with knowledge graphs is best understood as a difference of emphasis, not as an impossibility claim.
A bare knowledge graph primarily records entities and relations; entailment procedures, provenance, costs, and
verification policies may certainly be layered on top.  Positive Topology contributes a native language for cover,
positivity, witness preservation, and -- in the proposed extension -- feasible refinement.

\begin{quote}
A knowledge graph is naturally a representation of structured facts.  The PosTop programme aims to add a grammar of
refinement: what follows, what remains witnessable, and eventually what can be verified within available resources.
\end{quote}

\section{Software and Reproducibility}\label{app:software}

The LaTeX file you are reading lives in the same repository as the reference implementation of the forcing matrix, adjunctions, and notebooks.
The codebase is organized as follows:
\begin{itemize}
  \item \texttt{src/postop/core.py} implements the data structures discussed in Sections~\ref{sec:adj2}--\ref{sec:info-games}, with the forcing relation encoded as a Python dictionary of point $\to$ observable sets, plus the cost-aware variants in \texttt{src/postop/costs.py}.
  \item \texttt{src/postop/operators.py} and \texttt{src/postop/llm.py} provide tracing/explainability helpers that surface the cover witnesses/counterexamples required by the diagnostic examples.
  \item \texttt{examples/} mirrors the running medical and guardrail scenarios; \texttt{tests/} exercises the adjunction identities via \texttt{pytest}.
  \item \texttt{docs/REPRODUCIBILITY.md} gives a copy/paste runbook matching the commands below.
\end{itemize}
The latest sources live at \url{https://github.com/mirco-mannucci/postop}; we snapshot paper+code together via the git tag \texttt{v0.1.0} (and subsequent releases) so the commands below always map to a reproducible commit/DOI pair.

\subsection{Environment setup}

We keep the environment lightweight (pure Python + pytest). From the repository root:
\begin{verbatim}
python3 -m venv .venv
source .venv/bin/activate
pip install -e .[dev]
\end{verbatim}

\subsection{Running the examples}

The forcing table is serialized as ordinary Python dictionaries and can be inspected directly. To reproduce the medical walkthrough and the illustrative guardrail example discussed in this survey:
\begin{verbatim}
python examples/medical.py
python examples/llm_guardrail.py
\end{verbatim}
Each script prints the $\ext/\Box/\Diamond/\rest$ 
results, the compatibility witnesses, and the guardrail counterexamples that correspond to the formal statements in the main text.

\subsection{Notebooks and tests}

The ActionList roadmap adds four short notebooks under \texttt{notebooks/}; they can be executed headlessly via:
\begin{verbatim}
jupyter nbconvert --to notebook --execute \
  notebooks/01_basics_ext_int_hit_sel_j.ipynb
\end{verbatim}
Swap in the remaining notebook filenames (medical diagnosis, costs, guardrails) to reproduce the figures in the extended version of the paper.

Finally, run the test suite to check the adjunction/proof obligations:
\begin{verbatim}
pytest -q
\end{verbatim}
These tests import the installed package so they are venv-friendly and catch regressions both in the mathematical invariants (idempotence/contractivity) and in the explainability helpers referenced by the notebooks.



\begin{thebibliography}{99}

\bibitem{Sambin-OUP-PositiveTopology}
G.~Sambin.
\newblock \emph{Positive Topology: A New Practice in Constructive Mathematics}.
\newblock Oxford University Press, 2025.
\newblock \url{https://doi.org/10.1093/9780191746796.001.0001}.

\bibitem{CirauloSambin2012}
F.~Ciraulo and G.~Sambin.
\newblock A constructive {G}alois connection between closure and interior.
\newblock \emph{The Journal of Symbolic Logic}, 77(4):1308--1324, 2012.
\newblock \url{https://doi.org/10.2178/jsl.7704150}; arXiv:1101.5896.

\bibitem{Valentini2012}
S.~Valentini.
\newblock Relative formal topology: the binary positivity predicate comes first.
\newblock \emph{Mathematical Structures in Computer Science}, 22(1):69--102, 2012.
\newblock \url{https://doi.org/10.1017/S0960129511000466}.

\bibitem{CirauloVickers2016}
F.~Ciraulo and S.~Vickers.
\newblock Positivity relations on a locale.
\newblock \emph{Annals of Pure and Applied Logic}, 167(9):806--819, 2016.
\newblock \url{https://doi.org/10.1016/j.apal.2016.04.009}.

\bibitem{MacLaneMoerdijk1992}
S.~Mac~Lane and I.~Moerdijk.
\newblock \emph{Sheaves in Geometry and Logic: A First Introduction to Topos Theory}.
\newblock Springer-Verlag, New York, 1992.
\newblock \url{https://doi.org/10.1007/978-1-4612-0927-0}.

\bibitem{CoquandSambinSmithValentini2003}
T.~Coquand, G.~Sambin, J.~Smith, and S.~Valentini.
\newblock Inductively generated formal topologies.
\newblock \emph{Annals of Pure and Applied Logic}, 124(1--3):71--106, 2003.
\newblock \url{https://doi.org/10.1016/S0168-0072(03)00052-6}.

\bibitem{MannucciSambinForthcoming}
A.~Meneghello, M.~A.~Mannucci, and G.~Sambin.
\newblock \emph{Resource-Bounded Positive Topology: Threshold Forcing, Feasible Refinement, and the TROPOS Programme}.
\newblock In preparation, 2026.

\end{thebibliography}
\end{document}